\documentclass[runningheads]{llncs}
\usepackage{graphicx}
\usepackage{wrapfig}
\usepackage{subcaption}
\usepackage{float}
\usepackage{amsmath}
\usepackage{mathtools}
\usepackage[english]{babel}
\usepackage{cite}
\usepackage{textcmds}
\usepackage[utf8]{inputenc}
\usepackage{lineno}

\begin{document}
\title{Modularization, Composition, and Hierarchization of Petri Nets with {\normalfont \textsc{Heraklit}}}
\titlerunning{Modularization, Composition, and Hierarchization of Petri Nets}
%
%

\author{Peter Fettke\inst{1,2}\orcidID{0000-0002-0624-4431} \and
Wolfgang Reisig\inst{3}\orcidID{0000-0002-7026-2810}}
\authorrunning{P. Fettke, W. Reisig}
%

\institute{German Research Center for Artificial Intelligence (DFKI), Saarbr\"ucken, Germany \\
\email{peter.fettke@dfki.de}\\ \and
Saarland University, Saarbr\"ucken, Germany \\ \and
Humboldt-Universität zu Berlin, Berlin, Germany \\  
\email{reisig@informatik.hu-berlin.de}}

\maketitle              
\begin{abstract}
It is known for decades that computer-based systems cannot be understood without a concept of modularization and decomposition. We suggest a universal, expressive, intuitively attractive composition operator for Petri nets, combined with a refinement concept and an algebraic representation of nets and their composition. Case studies show exemplarily, how large systems can be composed from tiny net snippets. In the future, more field studies are needed to better understand the consequences of the proposed ideas in the real world.

\keywords{systems composition \and data modelling \and behaviour modelling \and composition calculus \and algebraic specification \and Petri nets \and Systems Mining}
\end{abstract}

\section*{Introduction}

Small Petri net models can be very attractive because their graphical representation is intuitively appealing. However, when models grow into nets with many places and transitions, their graphical representations quickly become unhandsome. Some kind of structure, abstraction, and composition of smaller nets are required. During recent decades, manifold proposals for composition and abstraction of Petri nets have been published. 

In this paper, we suggest a concept of Petri net \textit{modules}, their composition and their abstraction. This concept is particularly flexible and expressive. Notably,

\begin{enumerate}

\item an interface of a module may consist of a mixture of places, transitions, and arrows;

\item modules can be defined on the level of operational behavior as well as on abstract, hierarchical levels; 

\item several instances of a net $N$ can be composed: just write $N \bullet N$;

\item modules of different levels of abstraction can be composed;

\item composition defines a monoid on Petri net modules.
 
\end{enumerate}

Sec.~\ref{sec:informal_survey} provides an informal survey on the key aspects and properties of modules and their composition. Sec.~\ref{sec:formal_framework} provides the formal framework. This section may be skipped at first reading as the following case studies and examples in Sec.~\ref{sec:five_phils} and Sec.~\ref{sec:production_line} are intuitively easy to grasp.  The two case studies in these sections show best practice of using the module concept to specify big Petri net models by starting from small snippets of Petri nets, and composing them appropriately. The paper closes with a discussion of related work and some concluding remarks.

\section{An Informal Survey\label{sec:informal_survey}}

It is intuitively most attractive to start a theory or a conceptualization of the notion of \textit{discrete system} with the idea that a system is composed from subsystems. We follow this observation with a concept called \textit{modules}. To capture the composition of modules, it is furthermore a nearby idea that a module has \textit{elements}, some of which belong to the \textit{inner part} of the module, and the others to its \textit{interface}. Interface elements are \textit{labeled}, and the composition of two modules A and B results in a module $A \bullet B$, where equally labeled elements of the interfaces of A and of B have been merged, yielding an inner element of $A \bullet B$.  Unfortunately, this composition operator $\bullet$ is \textit{not associative}, i.e., for three such modules $A$, $B$ and $C$, in general, the module $(A \bullet B) \bullet C$  differs from  $A \bullet (B \bullet C)$. (As an example, assume equally labeled elements of all three interfaces of $A$, $B$ and $C$). But associativity is mandatory for cases where \textit{many} modules are to be composed: associative composition allows to cut a whole system into two sub-systems at any arbitrary point. If composition is not associative then a decomposition of a system is only possible inside the scope of a bracketing. 

We nevertheless stick to this concept of composition, but apply it in a specific manner, based on the observation that the interface of a module frequently decomposes quite naturally into two sub-interfaces. For example, a function has inputs and outputs, a specification has assumptions and guarantees, a manufacturer has material and products, a retailer has a buy side and a sell side, a firm has debts and liabilities, an income statement consists of profit and loss, a service has providers and requesters, an information system has push- and a pull operations, an agent has sensors and actors, a digital computer has input and output devices.  We abstract all this away and assume modules with \textit{two} interfaces, called \textit{left} and \textit{right}. 

For two modules $A$ and $B$, the composed module $A \bullet B$ is gained by merging each element of the right interface of $A$ with an equally labeled element of the left interface of $B$, yielding an inner element of $A \bullet B$. The remaining elements of the interfaces of $A$ and of $B$ go to the interfaces of $A \bullet B$. Most important, this kind of composition is associative, i.e. for any such modules $A$, $B$ and $C$ it is always the case that they can be composed, and that $(A \bullet B) \bullet C = A \bullet (B \bullet C)$. This implies the important observation that the composition $A_1 \bullet \dots \bullet A_n$ of \textit{many} modules $A_1, \dots, A_n$ can be written without brackets. 

In addition to composition, a second operator on modules, the \textit{closure} $A^c$, of a module $A$, yields cyclic structures.

This module concept comes with a useful notion of \textit{abstraction}: a module can be represented in any context just by its name. In particular, large specifications, formed $A_1 \bullet \dots \bullet A_n$, can be given short names and be employed later on by just calling this name. Vice versa, each module name can be expanded by its “refined” module. With this kind of abstraction, it is possible to easily define big and hierarchically structured modules. In technical terms, the abstraction of a module is again a module. Composition of abstract modules shows the architecture of a big model.

\section{The Formal Framework\label{sec:formal_framework}}

The reader may skip this section at first reading; in fact, the examples in the forthcoming sections should be self-explanatory.

\subsection{The Notion of Module\label{sec:2.1}}

In a systematic setting, we start with the assumption of a finite set $\Sigma$ of \textit{labels}. This in turn yields the notion of \textit{interfaces over} $\Sigma$. Technically, an interface is just a finite, labeled set, where each element has an index (a number). A \textit{module} is a graph, where two subsets of nodes serve as interfaces:

\begin{definition}[alphabet and interface]
Let $\Sigma$ be a finite set of symbols or labels, called an {\normalfont alphabet}. An {\normalfont interface over} $\Sigma$ is a finite set $R$, with each element of $R$ carrying a label of $\Sigma$, as well as an integer index, such that $n$ equally labeled elements of $R$ are indexed $1, \dots, n$. 
\end{definition}

Hence, all elements are indexed by $1$ if and only if no two elements are labeled alike. The elements are numerated $1, \dots, n$, if and only if $R$ has $n$ elements, all with the same label. 

Fig.~\ref{fig:interfaces} shows technical examples of two interfaces. An element $x$ may belong to more than one interface. An element always retains its label, but may have a different index in each interface.  For example, in Fig.~\ref{fig:interfaces}, the element $b$ has label $\beta$. Its index is $1$ in $R$ and index $2$ in $S$. The element $e$ with label $\alpha$ has index $3$ in $R$ and index $2$ in $S$.

\begin{figure}[!tb]
\centering
{\includegraphics[trim={0cm 0cm 0cm 0cm},clip,scale=.35]{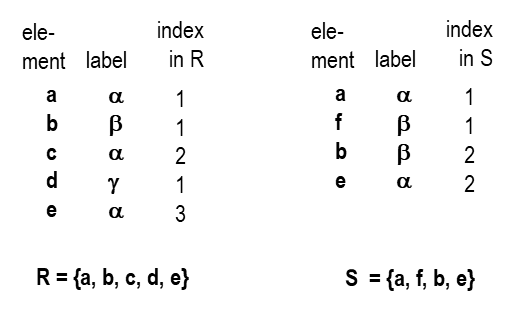}}
\caption{technical example for two interfaces $R$ and $S$, sharing the elements $a$, $b$, and $e$.}
\label{fig:interfaces}
\end{figure}

Notational convention for interfaces: We frequently represent an interface $R$ only by a column of its labels. This hides the individual identity of the elements.  The index of the elements of $R$ increase up-down. This way, the index of an element is given by the position of its label in the label column.

As usual, a directed graph $G$ is as a tuple $(V,E)$ of nodes $V$ and edges $E$; we conceive an undirected graph as a graph with edges in both directions. A module is a graph with two interfaces:

\begin{definition}[module]
Let $\Sigma$ be an alphabet, let $G = (V,E)$ be a graph, and let $^\ast G, G^\ast \subseteq V$ be two interfaces over $\Sigma$. Then G together with $^\ast G$ and $G ^\ast$ is a {\normalfont module over} $\Sigma$. 
\end{definition}

Notice that $^\ast G$ and $G ^\ast$ are not necessarily disjoint.

\textbf{Notation}. With the above definition, $^\ast G$ and $G ^\ast$ are the \textit{left and right interface of the module} $G$. Nodes not in an interface belong to the \textit{interior of} $G$, written $I_G$. {\ldq G\rdq} is the name of the module.

Graphical conventions guarantee unique representations of modules: 

\textbf{Graphical representation}. A module $G$ is graphically represented as usual for graphs: each node is represented as a dot and each arc as an arrow. For an undirected graph, arrows are replaced by lines. To represent $G$ with its interfaces, the inner of $G$ is surrounded by a box, with the elements of the left and the right interface on the left and the right margin, respectively. Fig.~\ref{fig:modules} shows examples: the inner elements of $A$ are the nodes $a$, $b$, and $c$; the inner elements of $B$ are $d$ and $e$. The interfaces follow the above notational conventions: The identity of the interface elements is hidden, only their label is presented. Furthermore, the indices of interface elements increase up-down. For example, $A ^\ast$ has one node lableled $\delta$ and two nodes labeled $\gamma$. a The $\gamma$-labeled node linked to the inner node $a$ has index $1$. The $\gamma$-labeled node linked to the inner nodes $a$ and $b$ has index $2$. The aim of lucid representations sometimes gives rise to further graphical conventions. They will be explained locally.

\begin{figure}[!tb]
\centering
{\includegraphics[trim={0cm 0cm 0cm 0cm},clip,scale=.18]{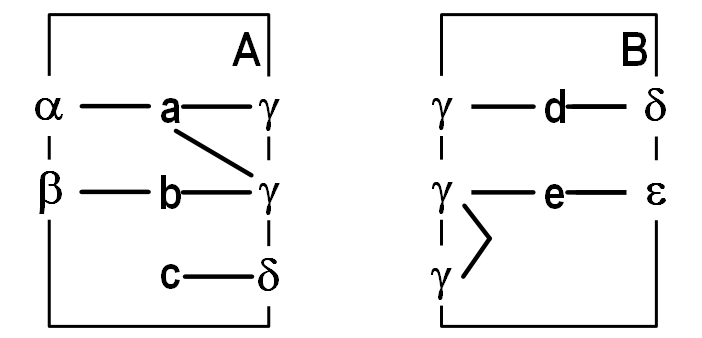}}
\caption{Two technical examples for modules. Interface elements are represented by their labels. Indices of equal labeled elements of an interface increase up-down.}
\label{fig:modules}
\end{figure}

\subsection{Composition of Modules\label{sec:2.2}}

A technical notion prepares the definition of the composition of modules: as already informally described in Sec.~\ref{sec:informal_survey}, composition $A \bullet B$ of two modules $A$ and $B$ essentially depends on the interfaces $A ^\ast$ and $^\ast B$. So, we need a definition that relates elements of two interfaces: 

\begin{definition}[harmonic pairs]
Let $R$ and $S$ be two disjoint interfaces. Two elements $r \in R$ and $s \in S$ are \textit{harmonic partners of $R$ and $S$} if and only if $r$ and $s$ have the same label, $l$, and the index $n$ of $r$ in $R$ coincides with the index of $s$ in $S$. Then $\{r, s\}$ is a {\normalfont harmonic pair with label $l$ and index $n$ of $R$ and $S$}.
\end{definition}

Fig.~\ref{fig:harmonic_pairs} outlines examples of harmonic pairs of the two interfaces of Fig.~\ref{fig:interfaces}.

\begin{figure}[!tb]
\centering
{\includegraphics[trim={0cm 0cm 0cm 0cm},clip,scale=.35]{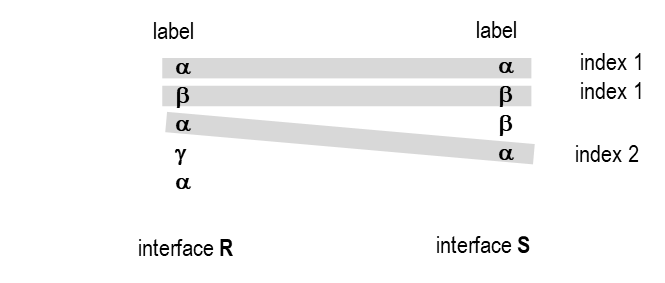}}
\caption{The thee harmonic pairs of the two interfaces $R$ and $S$ of Fig.~\ref{fig:interfaces}}.
\label{fig:harmonic_pairs}
\end{figure}

We are now prepared to define the composition of modules: 

\begin{definition}[$A \bullet B$]
Let $A$ and $B$ be two modules. For each node $x$ of $A$ or of $B$, let $x' = \{x, y\}$ if $\{x, y\}$ is a harmonic pair of $A ^\ast$ and $^\ast B$, and let $x' = x$ if no such pair exists. Then the module $A \bullet B$ is defined as follows:

\begin{enumerate}

\item The nodes of $A \bullet B$ are all $x'$ such that $x$ is a node of $A$ or of $B$.

\item The edges of $A \bullet B$ are all $(x', z')$, such that $(x, z)$ is an edge of $A$ or of $B$.

\item The left interface $^\ast(A \bullet B)$:

\begin{itemize}
\item For each $x \in {^\ast A}$ with label $l$ and index $n$ in $^\ast A$ holds:  $x' \in {^\ast (A \bullet B)}$ with label $l$ and index $n$ in $^\ast (A \bullet B)$. 
\item Let $x \in {^\ast B}$ without a harmonic partner in $A ^\ast$. Let $l$ be the label of $x$ and let $n$ be the index of $x$ in $^\ast B$. Let $p$ be the maximal index of $l$-labeled elements of $^\ast A$, and let $m$ be the number of $l$-labeled harmonic pairs of $A ^\ast$ and $^\ast B$. Then $x \in {^\ast (A \bullet B)}$ with label $l$ and index $p+n-m$. 
\end{itemize}

\item  The right interface $(A \bullet B) ^\ast$: 

\begin{itemize}
\item For each $x \in B ^\ast$ with label $l$ and index $n$ in $B ^\ast$ holds:  $x' \in (A \bullet B) ^\ast$ with label $l$ and index $n$ in $(A \bullet B) ^\ast$. 
\item Let $x \in A ^\ast$ without a harmonic partner in $^\ast B$. Let $l$ be the label of $x$ and let $n$ be the index of $x$ in $A ^\ast$. Let let $p$ be the maximal index of $l$-labeled elements of $^\ast A$, and let $m$ be the number of $l$-labeled harmonic pairs of $A ^\ast$ and $^\ast B$. Then $x \in (A \bullet B) ^\ast$ with label $l$ and index $p+n-m$. 
\end{itemize}

\end{enumerate}

\end{definition}

Hence, $x$ is an inner node of $A \bullet B$ if and only if $x$ is an inner node of $A$, or an inner node of $B$, or a harmonic pair of $A ^\ast$ and $^\ast B$.

Fig.~\ref{fig:composition} shows a technical example. For clarification, each interface node of the modules in this figure is equipped by its index. The inner nodes of $A \bullet B$ are the inner nodes $a, \dots, e$ of $A$ and $B$, as well as the two harmonic pairs of $A ^\ast$ and $^\ast B$, both labeled $\gamma$, with indices $1$ and $2$. In $A \bullet B$, those harmonic pairs are represented by $\{\gamma, \gamma\}_1$ and $\{\gamma, \gamma\}_2$, respectively.

\begin{figure}[!tb]
\centering
{\includegraphics[trim={0cm 0cm 0cm 0cm},clip,scale=.18]{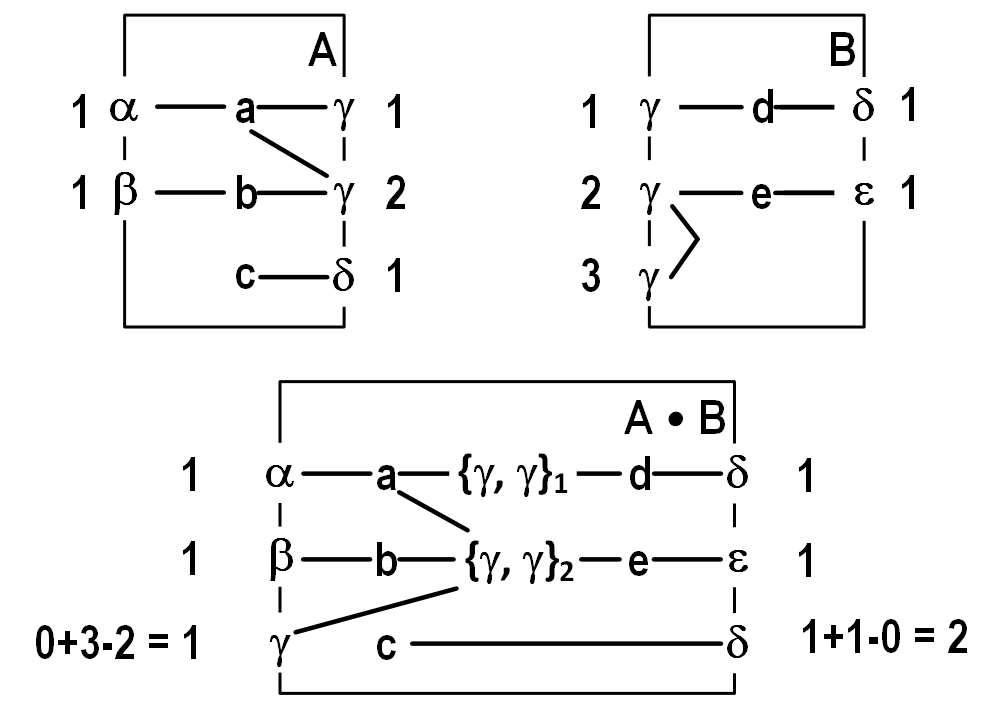}}
\caption{Composition of modules. Each interface element is equipped by its index.}
\label{fig:composition}
\end{figure}

Most important, the above composition operator is associative:

\begin{theorem}[associativity]
Let $A$, $B$, and $C$ be three modules over an alphabet. Then it holds that $(A \bullet B) \bullet C = A \bullet (B \bullet C)$.
\end{theorem}

Proof of this theorem is given in \cite{reisig2019associative}.

\subsection{The Closure Operator\label{sec:2.3}}

The following case studies show that the above composition operator is powerful enough to systematically construct big net models from tiny net snippets. Nevertheless, the operator does not cover the merge of harmonic partners of $A ^\ast$ and $^\ast A$. We express this by another, unary closure operator.

\begin{definition}[$A^c$]
Let $A$ be a module. For each node $x$ of $A$, let $x' = \{x,y\}$, if {x,y} is a harmonic pair of $A ^\ast$ and $^\ast A$, and let $x' = x$ if no such pair exists. Then the module $A^c$ is defined as follows:

\begin{enumerate}

\item The nodes of $A^c$ are all $x'$ such that $x$ is a node of $A$.

\item The edges of $A^c$ are all $(x', z')$, such that $(x, z)$ is an edge of $A$.

\item The left interface $^\ast (A^c)$: 
Let $x \in {^\ast A}$ without a harmonic partner in $A ^\ast$. Let $l$ be the label of $x$ and let $n$ be the index of $x$ in $^\ast A$. Let $m$ be the number of $l$-labeled harmonic pairs of $^\ast A$ and $A ^\ast$. Then $x \in {^\ast (A^c)}$ with label $l$ and index $n-m$. 

\item The right interface $(A^c) ^\ast$: 
Let $x \in A ^\ast$ without a harmonic partner in $^\ast A$. Let $l$ be the label of $x$ and let $n$ be the index of $x$ in $A ^\ast$. Let $m$ be the number of $l$-labeled harmonic pairs of $A ^\ast$ and $^\ast A$. Then $x \in (A^c) ^\ast$ with label $l$ and index $n-m$. 

\end{enumerate}

\end{definition}

Fig.~\ref{fig:closure} shows a technical example: each interface element carries its index. The inner elements of $A^c$ are the inner elements $a, \dots, g$ of $A$ as well as the two harmonic pairs of $A ^\ast$ and $^\ast A$, labeled $\alpha$ and $\beta$, both with index $1$. In $A^c$, they are represented by $\{\alpha, \alpha\}$ and $\{\beta, \beta\}$.

\begin{figure}[!tb]
\centering
{\includegraphics[trim={0cm 0cm 0cm 0cm},clip,scale=.18]{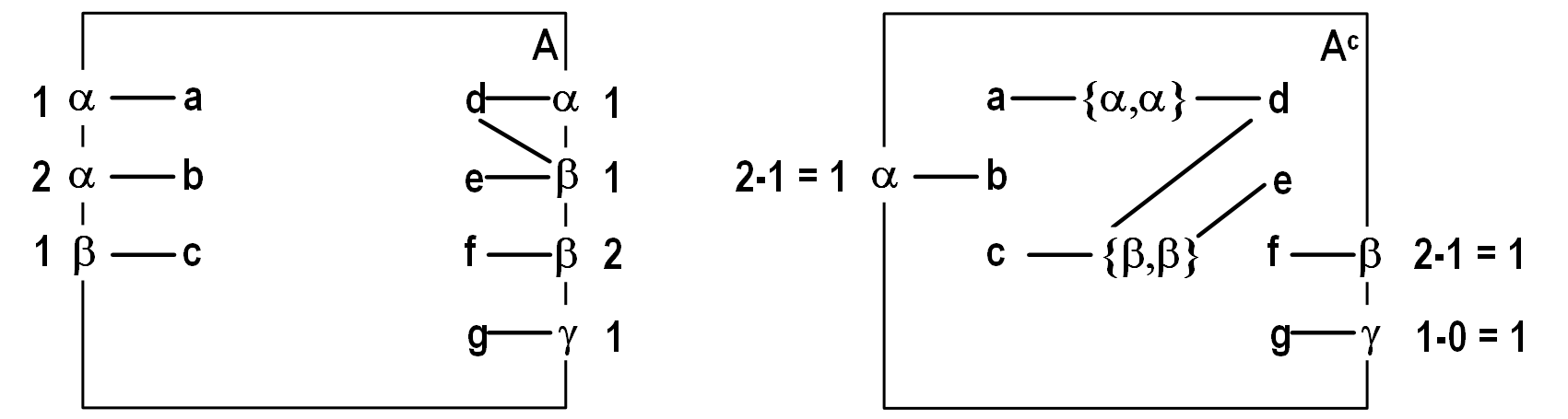}}
\caption{technical example: a module and its closure}
\label{fig:closure}
\end{figure}

The closure operator is \textit{idempotent}: for each module $A$ holds: $(A^c)^c  =  A^c$. This observation is based on the fact that no label occurs in both the left as well as the right interface of $A^c$.

\subsection{Refinement and Hierarchies\label{sec:2.4}}

Composition $A \bullet B$ of two net modules $A$ and $B$ returns a net module from which $A$ and $B$ can in general not be retained. Sometimes, however, it is useful to keep information about modules that were intermediately constructed during the design process of a big net module. This is achieved by abstract modules. The nodes of an \textit{abstract module} are names of other modules. 

In technical terms, an abstract module is composed from \textit{atomic} modules.  An atomic module has \textit{exactly one} inner node. This node is linked to each interface node. Each module $A$ has its abstract version, $abstr(A)$, which is an atomic module that inherits the left and the right interface from $A$, but abstracts all inner details of $A$ away, replacing them by the name of $A$. 

Fig.~\ref{fig:abstract_versions} shows the abstract versions $abstr(A)$ and $abstr(B)$ of the modules $A$ and $B$ of Fig.~\ref{fig:modules}, together with an obvious shorthand representation. The composition $abstr(A) \bullet abstr(B)$ of the abstract versions of $A$ and $B$ is not atomic, as Fig.~\ref{fig:abstract_and_composition} shows.

\begin{figure}[!tb]
\centering
{\includegraphics[trim={0cm 0cm 0cm 0cm},clip,scale=.17]{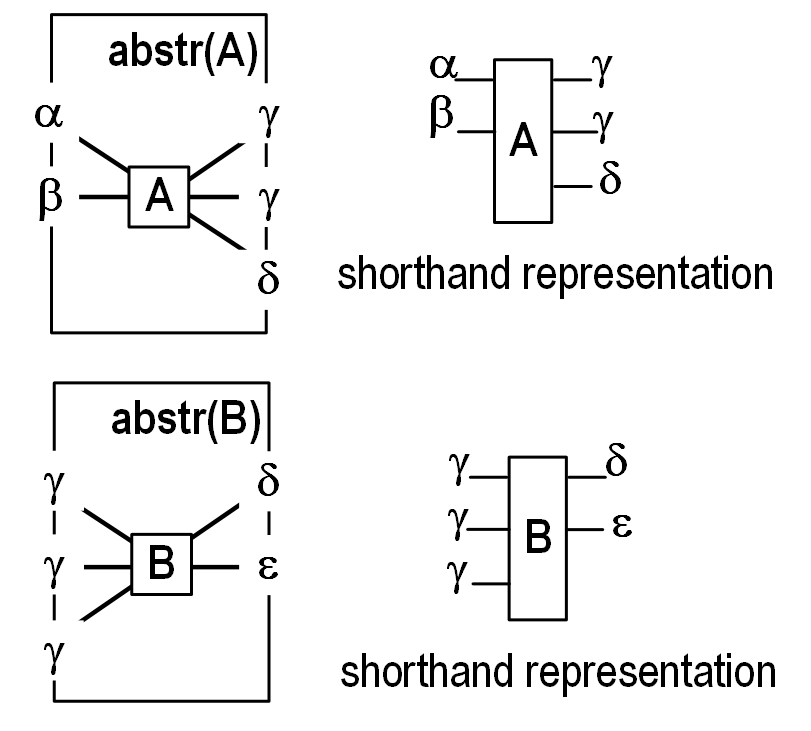}}
\caption{the abstract versions of the modules of Fig.~\ref{fig:modules}}
\label{fig:abstract_versions}
\end{figure}

\begin{figure}[!tb]
\centering
{\includegraphics[trim={0cm 0cm 0cm 0cm},clip,scale=.17]{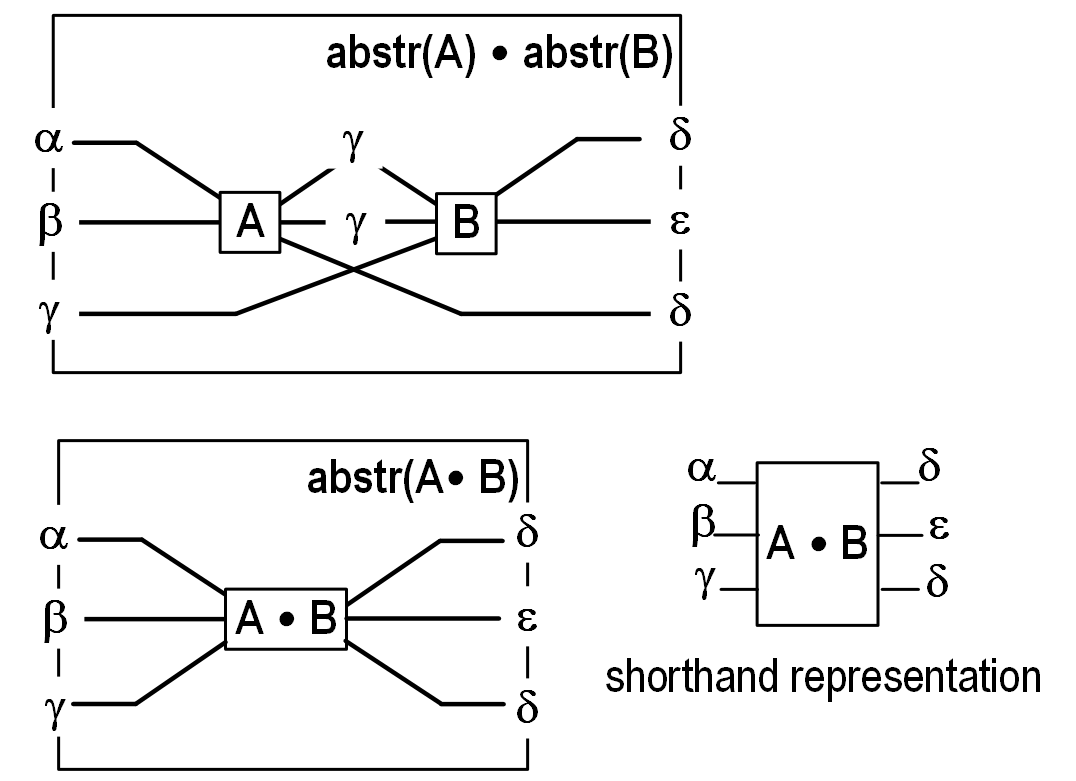}}
\caption{composition and abstraction}
\label{fig:abstract_and_composition}
\end{figure}

The operator $abstr$ has some interesting algebraic properties: For all modules $A$ and $B$ holds, up to re-naming atomic modules:
\begin{itemize}
\item $abstr(abstr(A)) = abstr(A)$;
\item $abstr(A \bullet B) = abstr( abstr(A) \bullet abstr(B))$. 
\end{itemize}

Frequently, a module $A$ is composed from given and well-known modules $A_1, \dots, A_n$. Then for $A = A_1 \bullet \dots \bullet A_n$, the abstract module $B = abstr(A_1) \bullet \dots \bullet abstr(A_n)$ is the \textit{seam} of $A$.

\subsection{The Module Monoid\label{sec:2.5}}

An extreme example of a module is the \textit{empty} module, $E$. It has no nodes at all, and, of course, its interfaces are empty. For any module $A$ then holds according to the above definitions: $E \bullet A = A \bullet E = A$. 

For a given alphabet $\Sigma$, the set $M_\Sigma$ of all modules over $\Sigma$ forms an algebraic structure $\underline{M_\Sigma} \coloneqq (M_\Sigma, \bullet, E)$, called a \textit{monoid}. This means that composition is \textit{total}, i.e. any two modules $A$ and $B$ in $M_\Sigma$ can be composed, and the composition $A \bullet B$ is again a module in $M_\Sigma$. Furthermore, composition is associative, as discussed at the end of Sec.~\ref{sec:2.2}. This provides a solid basis for intuitive and technically simple design of big models. 

The monoid $\underline{M_\Sigma}$ resembles the monoid $(\Sigma ^\ast, \bullet, \epsilon)$ of words over $\Sigma$ (where the composition symbol is usually skipped). Consequently, the theory of formal languages is applicable to modules. 

The module monoid $\underline{M_\Sigma}$ has interesting submonoids, viz. subsets such that composition of modules in the subset yields again a module in this subset. Most important, Petri nets form such a submonoid:  a Petri net $N = (P, T; F)$ defines a graph $(P \cup T, F)$: each place and each transition of $N$ is a node, and the flow relation of $N$ contributes the arcs of the graph. The alphabet $\Sigma$ is partitioned into two sets $\Sigma _P$ and $\Sigma _T$, such that the label of each place is taken from $\Sigma _P$, and the label of each transition from $\Sigma _T$. This guarantees that the composition of Petri nets is a Petri net again.

Likewise important is the submonoid of abstractions, as defined above. The $\Sigma$-labeled occurrence nets are a submonoid of the Petri net monoid. 

As mentioned in Sec.~\ref{sec:2.1} already, the two interfaces $^\ast G$ and $G ^\ast$ of a module $G$ are not necessarily disjoint. In the extreme case of $^\ast G = G ^\ast$, the module $G$ is called \textit{monolithic}. The monolithic modules also constitute a submonoid of $\underline{M_\Sigma}$. This observation will be used in the next section.

\subsection{Completeness of Composition and Closure\label{sec:2.6}} 

Here we show that every net structure can be generated by help of composition from \textit{transition atoms}: A transition atom $A$ is an atomic module as defined in Sec.~\ref{sec:2.4}, with a transition $t$ as the only inner element, and its joined pre- and postsets ${^\bullet t} \cup t ^\bullet$  as its left as well its right interface; for the sake of simplicity we assume nets without isolated elements.

\begin{definition}
Let $N = (P, T; F)$ be a net.

\begin{enumerate}

\item For $t \in T$, let the module $[t]$ be the graph $(V, E)$ with $V = {^\bullet t} \cup t ^\bullet \cup \{t\}$ and $E = (^\bullet t \times {t}) \cup ({t} \times t ^\bullet)$. The interfaces $^\ast [t]$ and $[t] ^\ast$ of $[t]$ are identical, and $^\ast [t] = [t] ^\ast = {^\bullet t} \cup t ^\bullet$.  Each interface place is labeled by its identity. The module [t] is called the {\normalfont transition atom} of $t$ and is monolithic per construction.

\item Let the module $[N]$ be the graph $(V, E)$ with $V = P \cup T$ and $E = F$. Let $^\ast [N] = [N] ^\ast = P$.

\end{enumerate}

\end{definition}

\begin{lemma} Let $N = (P, T; F)$ be a net with $T = {t_1, \dots, t_n}$.  Let $A_0 = E$, and for $i = 1, \dots, n$ let $A_i = A_{i-1} \bullet [t_i]$ . Then $[N] = A_n$.

\end{lemma}

\begin{proof}
by induction on $|T|$: 

The case $|T| = 0$ is trivial. 

Assume the Lemma holds for $n-1$. Let $T' = \{t_1, \dots , t_{n-1}\}$; $P' = {^\bullet t_1} \cup t_1 ^\bullet \cup \dots \cup {^\bullet t_{n-1}} \cup t_{n-1} ^\bullet$; $F' =  F \cap ((P' \cup T') \times (P' \cup T')$, and  let $N' =  (P',T',F')$.  
Then $[N] = [N'] \bullet [t_n]$, by construction of $N'$. 
Then $[N] = A_{n-1} \bullet [t_n]$, because $[N'] = A_{n-1}$, by inductive assumption.
Then $[N] = A_n$,  by definition of $A_n$.

\end{proof}
	 
Notice that not all graphs can be constructed from finitely many transition atoms. Examples are cycles with uneven number of nodes.

\section{Case Study: The Five Philosophers System\label{sec:five_phils}} 

We assume the reader be aware of Dijkstra’s paradigm of five philosophers, sitting round a table \cite{dijkstra1971hierarchical}. Neighboring philosophers share a fork. Hence, each philosopher has a left and a right fork, and each fork has a left and a right user. To take a meal, a philosopher requires both his left and his right fork. 

We start with four quite simple snippets of nets, and embed each of them into a module. From these modules we construct larger modules and finally a Petri net model of the five philosophers system, just by help of the two operators for composition and closure. We achieve this in two different ways. The first one focusses on the involved forks, and composes the resulting module from modules for forks (and their users). The second way focusses on the philosophers and composes the same module from modules for philosophers (and their forks).

\subsection{The Modules of Forks and Philosophers\label{sec:3.1}}
We start with a model for a very basic piece of behavior of forks: a fork is taken by its left user and returned later on (after the user has finished his meal). Fig.~\ref{fig:four_snippets_a} shows this behavior as a simple sequence of two transition with an intermediate place, arranged in a module, named \textit{left use}: the transitions labeled \textit{take} and \textit{return} are in the left interface, the place labeled \textit{available} is in the right interface. Choice of the two interfaces will be motivated below.

\begin{figure}[!tb]
\centering
\subcaptionbox{left use of a fork\label{fig:four_snippets_a}}
{\includegraphics[trim={0cm 0cm 0cm 0cm},clip,scale=.21]{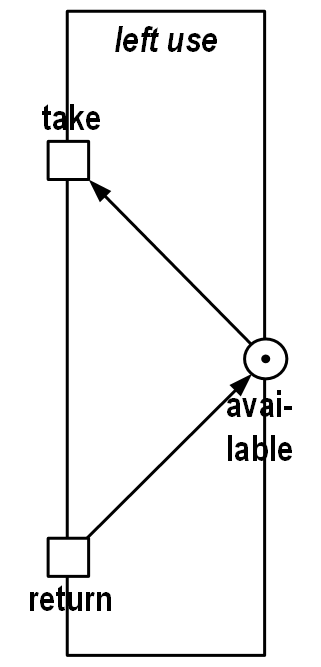}}\hspace{0.5cm}
\subcaptionbox{right use of a fork\label{fig:four_snippets_b}}
{\includegraphics[trim={0cm 0cm 0cm 0cm},clip,scale=.21]{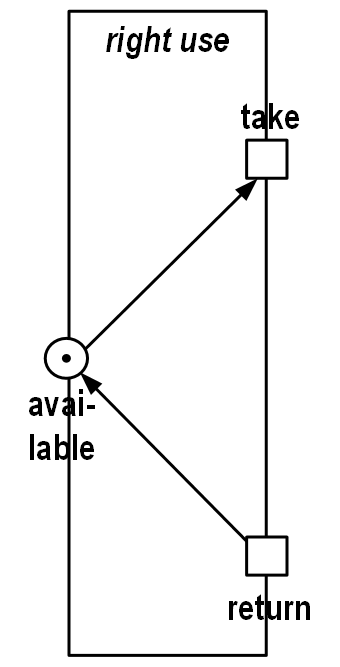}}\hspace{0.5cm}
\subcaptionbox{thinking philosopher\label{fig:four_snippets_c}}
{\includegraphics[trim={0cm 0cm 0cm 0cm},clip,scale=.21]{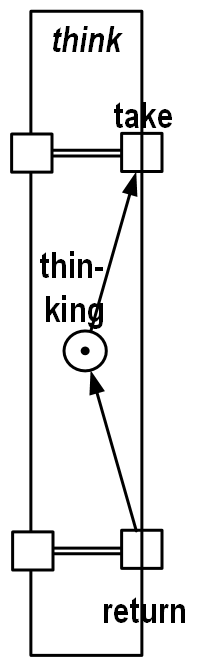}}\hspace{0.5cm}
\subcaptionbox{eating philosopher\label{fig:four_snippets_d}}
{\includegraphics[trim={0cm 0cm 0cm 0cm},clip,scale=0.21]{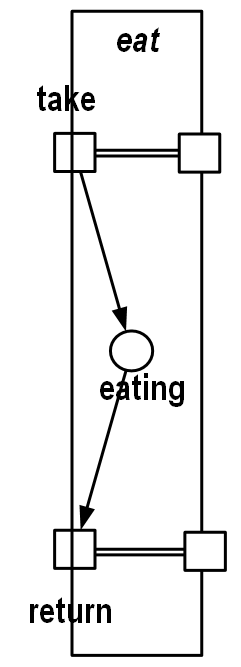}}
\caption{four tiny snippets of Petri nets}
\label{fig:four_snippets}
\end{figure}

In analogy to Fig.~\ref{fig:four_snippets_a}, the \textit{right use} module of four snippets in Fig.~\ref{fig:four_snippets_b} shows the case of a fork being taken and returned by its right user. The interfaces of this module mirror the interfaces of Fig.~\ref{fig:four_snippets_a}.

The module think in Fig.~\ref{fig:four_snippets_c} shows the behavior of a thinking philosopher: A thinking philosopher takes a fork, and returns it later on. This module resembles structurally the modules in Figs.~\ref{fig:four_snippets_a} an \ref{fig:four_snippets_b}. However, its elements are differently placed: Firstly, the place \textit{thinking} is located in the inner of the module. Secondly, for reasons to become clear later on, we wish the transition \textit{take} to be a member of both, the left and the right interface of the module \textit{think}.  This is conceptually simple, but graphically this is challenging. As a graphical convention, we include an additional graphical instance of the \textit{take} transition, placed on the left margin of \textit{think}, as in Fig.~\ref{fig:four_snippets_c}. A double line, resembling equality, indicates that both squares represent just one transition. Likewise, the transition \textit{return} belongs to the left as well as the right interface of \textit{think}. Finally Fig.~\ref{fig:four_snippets_d} shows the \textit{eat} module.

\subsection{The Forks-based Model\label{sec:3.2}}
The four modules of Fig.~\ref{fig:four_snippets} appear not too impressive. What makes them interesting is their \textit{composition}: In the philosophers system, a fork can be utilized by its left as well as its right user; any other behavior is not intended. Fig.~\ref{fig:forks_a} shows this behavior as a module, called fork. This module is gained as the composition $\textit{fork} \coloneqq \textit{left use} \bullet \textit{right use}$ of the modules \textit{left use} and \textit{right use} of Figs.~\ref{fig:four_snippets_a} and \ref{fig:four_snippets_b}.

\begin{figure}[!tb]
\centering
\subcaptionbox{$\textit{fork} \coloneqq
\textit{left use} \bullet \textit{right use}$\label{fig:forks_a}}
{\includegraphics[trim={0cm 0cm 0cm 0cm},clip,scale=.18]{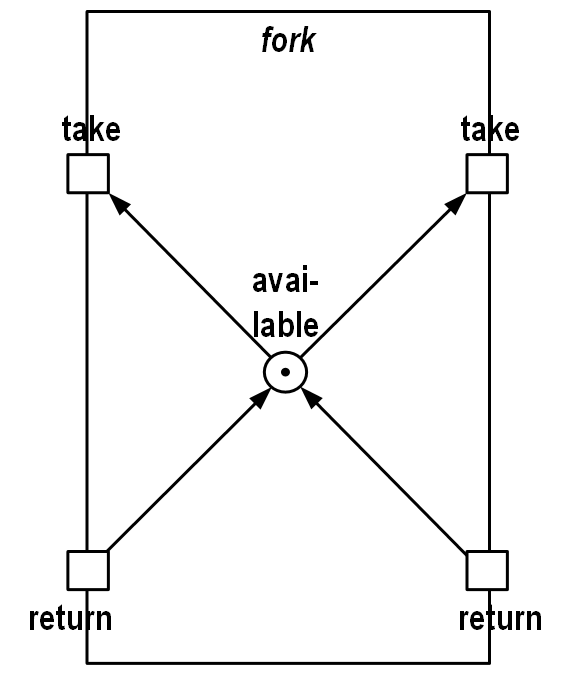}}\hspace{0.5cm}
\subcaptionbox{two representations of the module 
$\textit{fork with left user} \coloneqq \textit{think} \bullet \textit{fork}$\label{fig:forks_b}}
{\includegraphics[trim={0cm 0cm 0cm 0cm},clip,scale=.18]{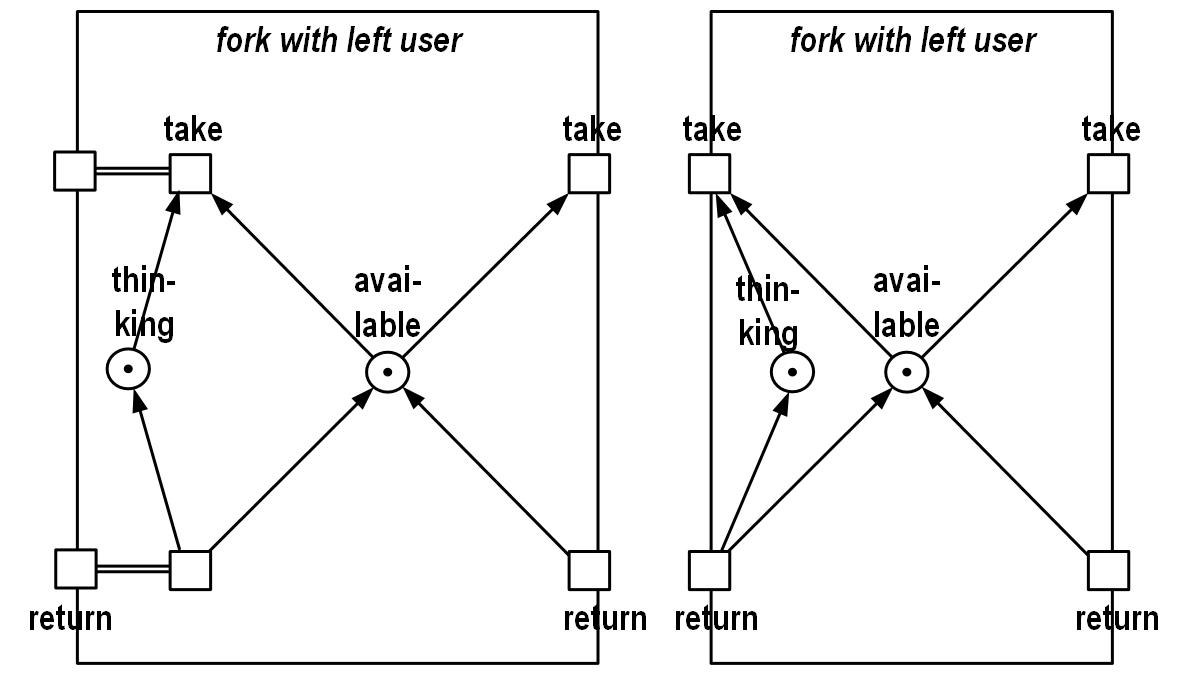}}
\caption{The modules of \textit{fork} and \textit{fork with left user}}
\label{fig:forks}
\end{figure}

\begin{figure}[!tb]
\centering
{\includegraphics[trim={0cm 0cm 0cm 0cm},clip,scale=0.25]{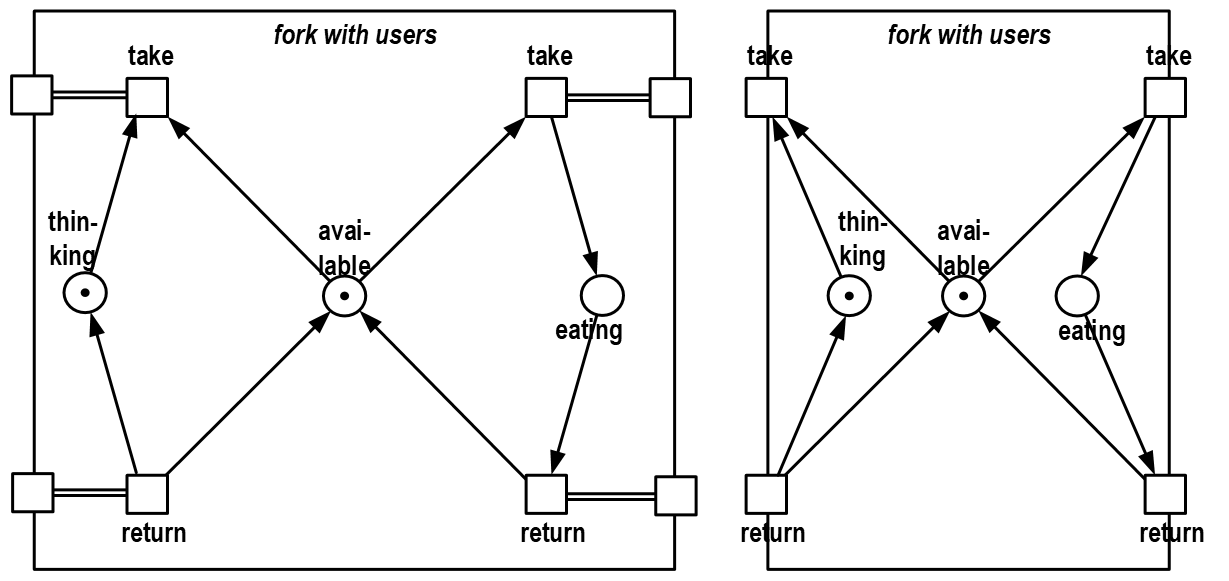}}
\caption{two representations of the module $\textit{fork with users} \coloneqq \textit{think} \bullet \textit{fork} \bullet \textit{eat}$}
\label{fig:forks_with_users}
\end{figure}

\begin{figure}[!tb]
\centering
\subcaptionbox{the module $\textit{forks in a row} \coloneqq \textit{fork with users} \bullet \textit{fork with users} \bullet \textit{fork with users} \bullet \textit{fork with users} \bullet \textit{fork with users}$\label{fig:forks_based_a}}
{\includegraphics[trim={0cm 0cm 0cm 0cm},clip,width=1.1\textwidth]{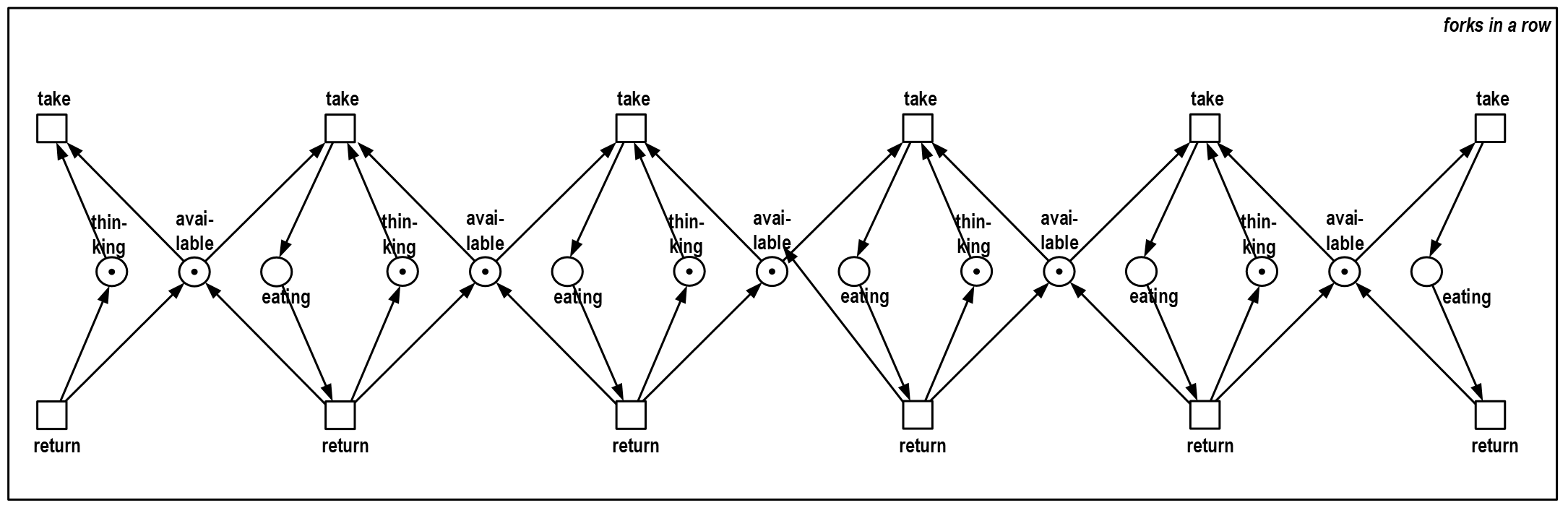}}\vspace{0.5cm}
\subcaptionbox{the module $\textit{forks in a circle} \coloneqq (\textit{forks in a row})^c$\label{fig:forks_based_b)}}
{\includegraphics[trim={0cm 0cm 0cm 0cm},clip,width=1.1\textwidth]{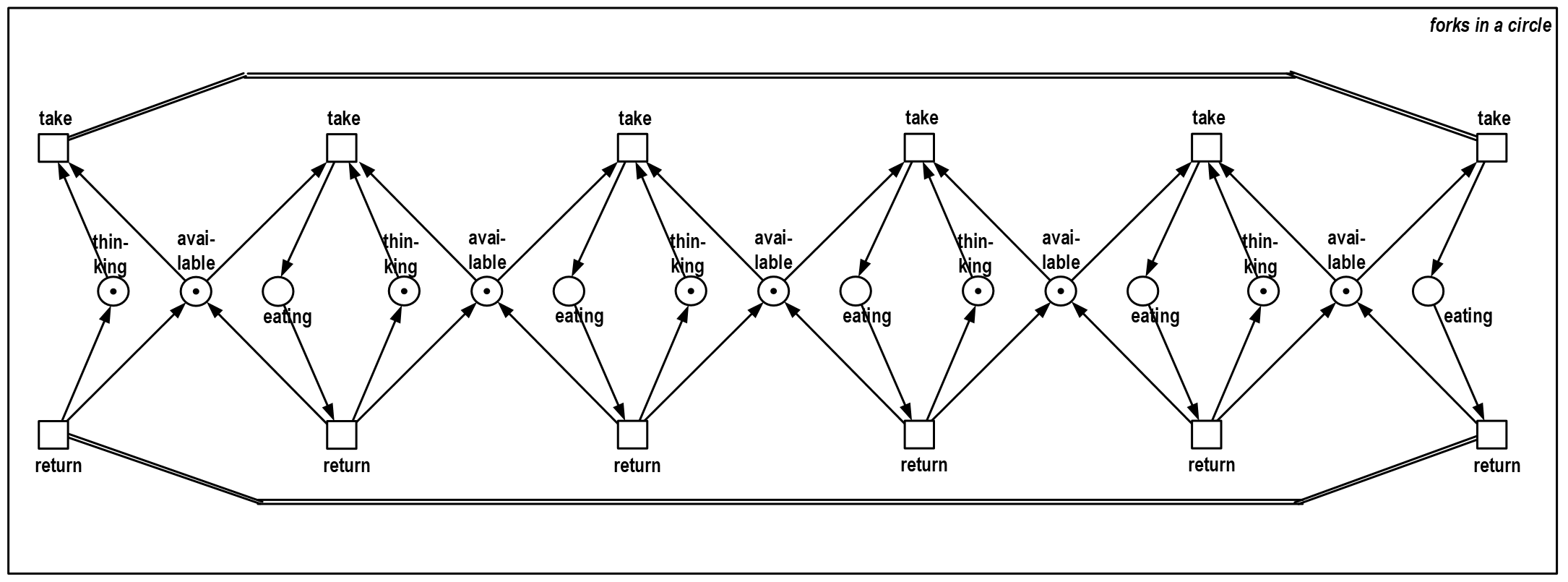}}
\caption{composing five forks}
\label{fig:forks_based}
\end{figure}

We continue by composing a thinking philosopher with the fork, as in Fig.~\ref{fig:forks_b}. Both graphical representations show the same module. Notice that we might have constructed $A \coloneqq \textit{think} \bullet \textit{left use}$ as an intermediate module, and then construct $A \bullet \textit{right use}$, resulting again in the module $\textit{think} \bullet \textit{fork}$ of Fig.~\ref{fig:forks_b}.

We further extend Fig.~\ref{fig:forks_b} by the right user, yielding the module \textit{fork with users} of Fig.~\ref{fig:forks_with_users}. We can now compose five copies of this module, thus gaining the module \textit{forks in a row} of Fig.~\ref{fig:forks_based_a}.  

Dijkstra’s philosophers do not sit in an row, but in a circle. To model this \textit{take} labeled transitions of the left and the right interface of the module \textit{forks in a row} of Fig.~\ref{fig:forks_based_a} must be merged. The \textit{return} labeled interface transitions must be merged likewise. This is easily achieved in the graphical representation: As in Figs.~\ref{fig:forks} and \ref{fig:forks_with_users}, a double line links different representations of one and the same element. So, the module \textit{forks in a cycle} of Fig.~\ref{fig:forks_based_b)} shows a perfect model of Dijkstra’s five philosophers. This module’s right and left interfaces both are empty. 

The step from Fig.~\ref{fig:forks_based_a} to Fig.~\ref{fig:forks_based_b)} cannot be represented by means of the composition operator! Instead, we need the closure operator, as introduced in Sec.~\ref{sec:2.3}.

Summing up, based on the four snippets of Fig.~\ref{fig:four_snippets}, we can define the philosophers’ system by purely algebraic constructs, using the module for forks:

\begin{equation}
\begin{split}
fork                       \coloneqq  & \; \textit{left use} \bullet \textit{right use} \\
\textit{fork with users} \coloneqq  & \; \textit{think} \bullet \textit{fork} \bullet \textit{eat} \\
\textit{forks in a row}   \coloneqq & \; \textit{fork with users} \bullet \textit{fork with users} \bullet \textit{fork with users} \; \bullet \\ & \;  \textit{fork with users} \bullet \textit{fork with users} \\
\textit{forks in a cycle}  \coloneqq & \; (\textit{fork in a row})^c
\end{split}
\end{equation}

Usual specifications of the five philosopher system employ forks or phils, numbered (or indexed) $1, \dots, 5$, and use addition modulo 5 to merge the leftmost and the rightmost philosopher – not too elegant a construction. The \textit{fork in a cycle} definition specifies a “perfect” cycle, without distinguishing philosophers individually.

\subsection{The Philosophers-based Model\label{sec:3.3}}

In Sec.~\ref{sec:3.2}, we modelled the five philosophers system from the perspective of forks: The module \textit{fork} of Fig.~\ref{fig:forks_a} has been complemented by the left and right user. Five instances of this extended module have then been composed to the complete model in Fig.~\ref{fig:forks_based_b)}.  

One can derive this model also from the perspective of the philosophers. Starting with the Petri net snippets of Fig.~\ref{fig:four_snippets}, the module \textit{phils} of Fig.~\ref{fig:phils_a}, composed from the snippets \textit{think} and \textit{eat}, models a single philosopher. Notice that the transition \textit{take} as well as the transition \textit{return} both belong to both, the left and the right interface. For the sake of symmetry and elegance, these transitions are also represented inside the module. Fig.~\ref{fig:phils_b} expands this module by the philosopher’s use of his left and right fork. In analogy to Fig.~\ref{fig:forks_based_a}, five instances of the philosopher module are composed in a row in Fig.~\ref{fig:phils_based_a}, and bent to a cycle in Fig.~\ref{fig:phils_based_b}. In fact, the two modules of Fig.~\ref{fig:forks_based_b)} and Fig.~\ref{fig:phils_based_b} are identical.

\begin{figure}[!tb]
\centering
\subcaptionbox{$phil \coloneqq \textit{think} \bullet \textit{eat}$\label{fig:phils_a}}
{\includegraphics[trim={0cm 0cm 0cm 0cm},clip,scale=.20]{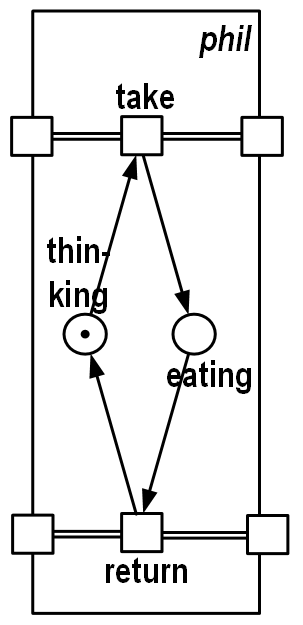}}\hspace{.5cm}
\subcaptionbox{two representations of the module
$\textit{phil with forks} \coloneqq \textit{left use} \bullet \textit{phil} \bullet \textit{right use}$\label{fig:phils_b}}
{\includegraphics[trim={0cm 0cm 0cm 0cm},clip,scale=.20]{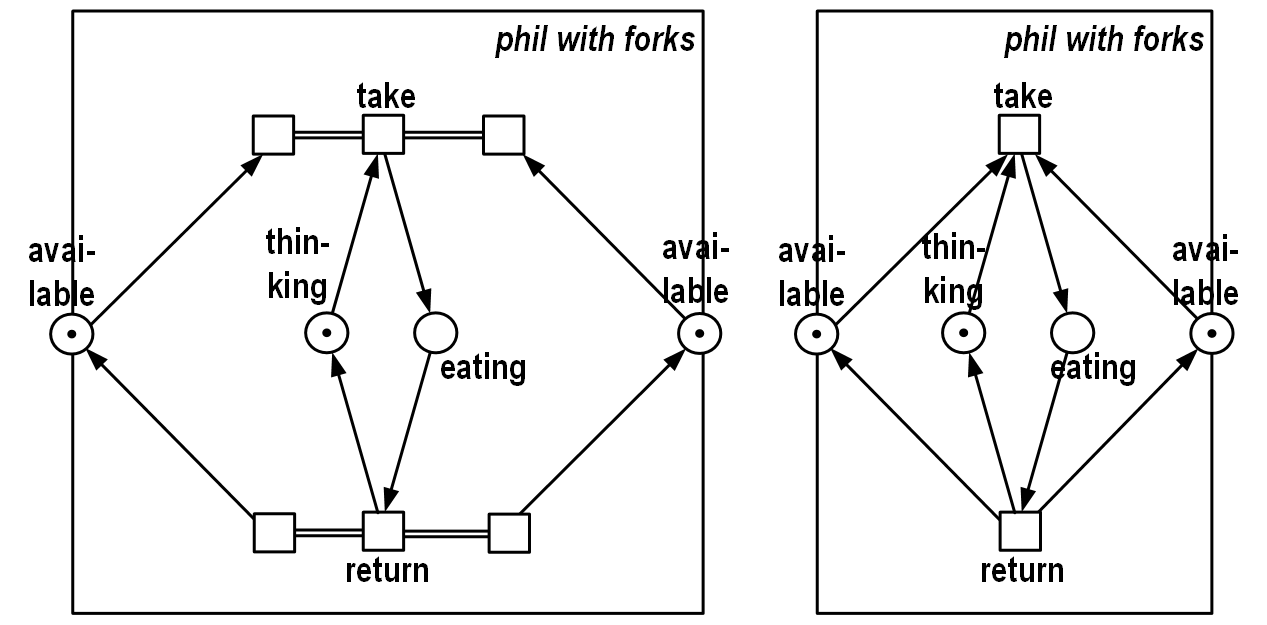}}
\caption{The modules of \textit{phil} and \textit{phil with forks}}
\label{fig:phils}
\end{figure}

\begin{figure}[!tb]
\centering
\subcaptionbox{the module $\textit{phils in a row} \coloneqq \textit{phil with forks} \bullet \textit{phil with forks} \bullet \textit{phil with forks} \bullet \textit{phil with forks} \bullet \textit{phil with forks}$\label{fig:phils_based_a}}
{\includegraphics[trim={0cm 0cm 0cm 0cm},clip,width=1\textwidth]{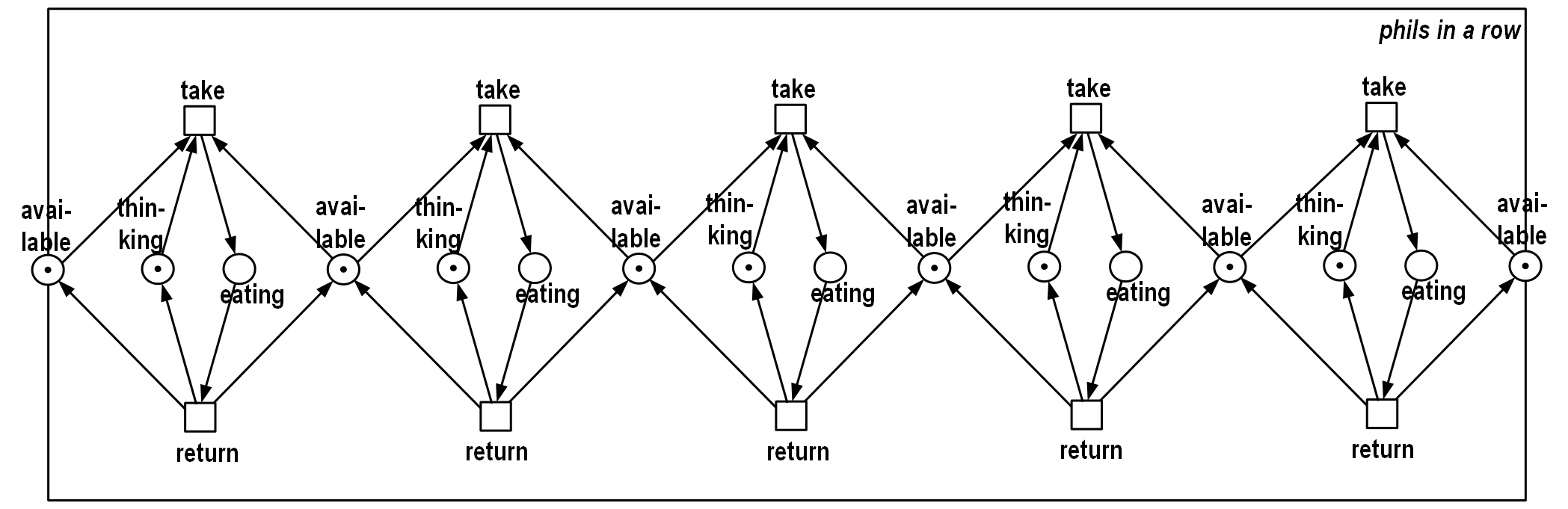}}\vspace{0.5cm}
\subcaptionbox{the module
$\textit{phils in a circle} \coloneqq (\textit{phils in a row}) ^\ast$\label{fig:phils_based_b}}
{\includegraphics[trim={0cm 0cm 0cm 0cm},clip,width=1\textwidth]{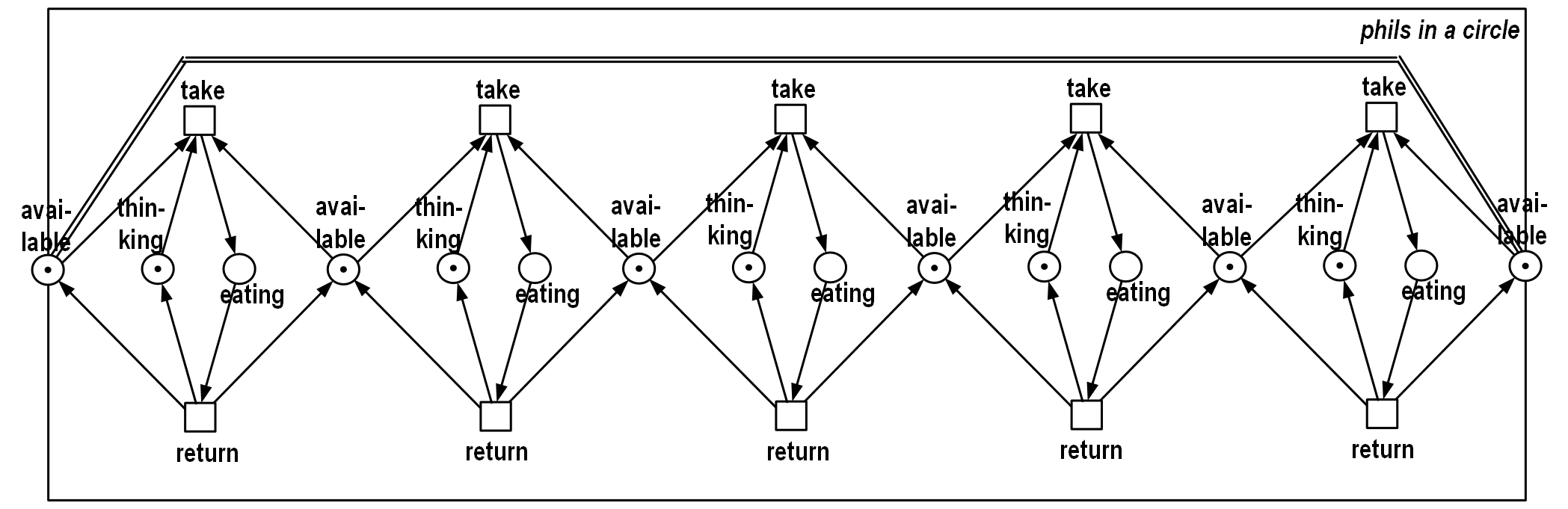}}
\caption{composing five phils}
\label{fig:phils_based}
\end{figure}

Again, based on the four snippets of Fig.~\ref{fig:four_snippets}, we can define the philosophers’ system by purely algebraic constructs, using the module for philosophers:

\begin{equation}
\begin{split}
\textit{phil} \coloneqq & \; \textit{think} \bullet \textit{eat} \\
\textit{phil with forks} \coloneqq  & \; \textit{left use} \bullet \textit{phil} \bullet \textit{right use} \\
\textit{phils in a row} \coloneqq  & \; \textit{phil with forks} \bullet \textit{phil with forks} \bullet \textit{phil with forks} \; \bullet \\ & \; \textit{phil with forks} \bullet \textit{phil with forks}  \\ 
\textit{phils in a cycle} \coloneqq  \; & (\textit{phils in a row})^c \\
\end{split}
\end{equation}

This definition of \textit{phils in a cycle} and \textit{forks in a cycle} from Sec.~\ref{sec:3.2} yield identical nets.

\subsection{Abstraction\label{sec:3.4}}
Composition yields “seamless” results: In general, it is not possible to re-compute modules $A$ and $B$ from $A \bullet B$. Nevertheless, occasionally one is interested in the history of a design process. For instance, one may ask whether the philosophers’ model has been generated via forks or via phils. This kind of information can be retained by means of \textit{abstractions}. Abstract representations of modules on any level can be composed and again be abstracted. Using the shorthand representations of Fig.~\ref{fig:abstract_versions}, Fig.~\ref{fig:abstract_forks} shows abstractions of the four behavioral snippets of Fig.~\ref{fig:four_snippets}, as well as abstract versions of forks and forks with users, etc. Fig.~\ref{fig:abstract_phils} shows corresponding abstractions of the phils-based path to the overall system. Concrete and abstract representations on any level of detail can be composed, as Fig.~\ref{fig:mixed_version}.

\begin{figure}[!tb]
\centering
{\includegraphics[trim={0cm 0cm 0cm 0cm},clip,width=1\textwidth]{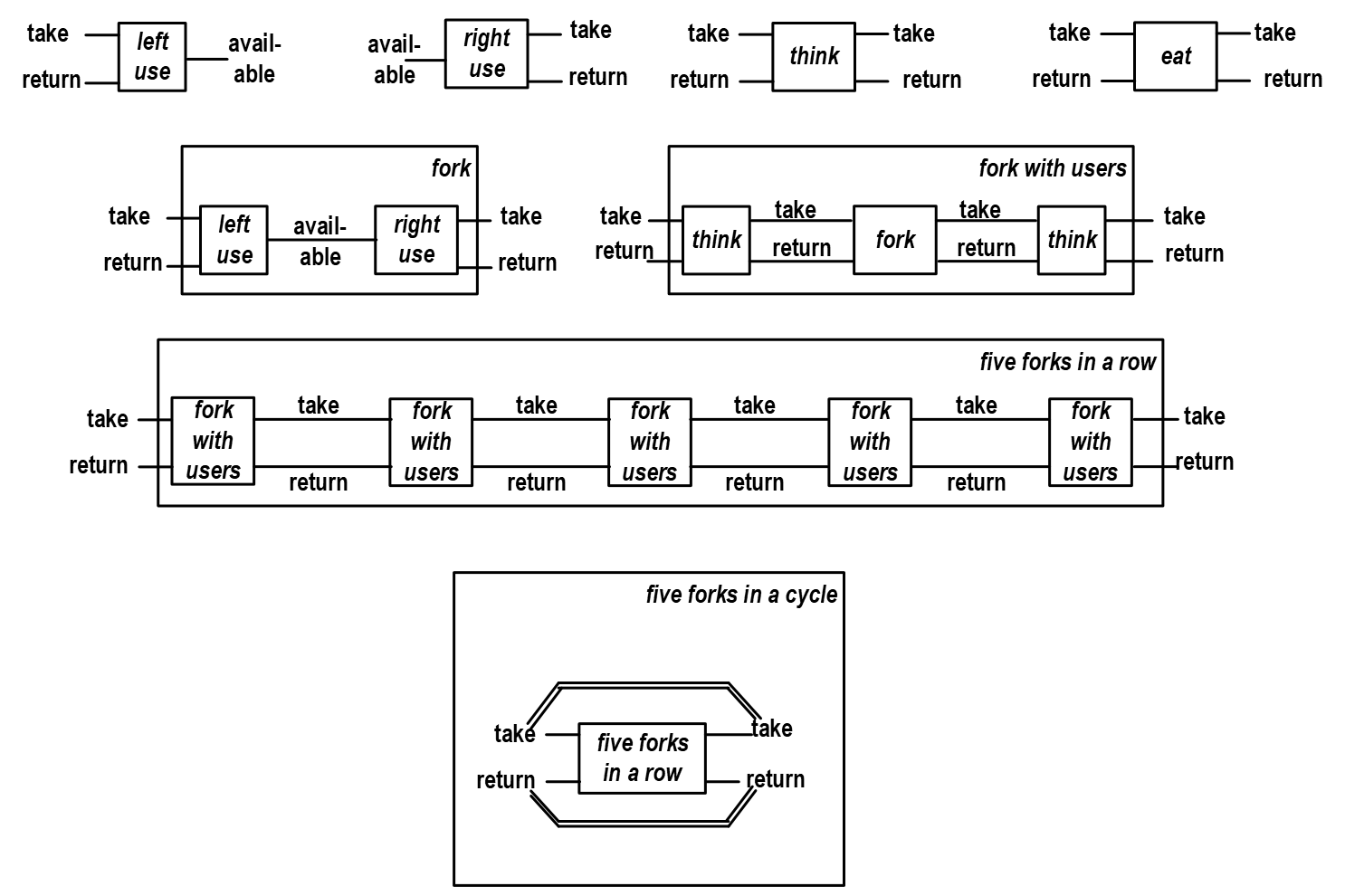}}
\caption{abstract versions of the \textit{forks}-based modules}
\label{fig:abstract_forks}
\end{figure}

\begin{figure}[!tb]
\centering
{\includegraphics[trim={0cm 0cm 0cm 0cm},clip,width=1\textwidth]{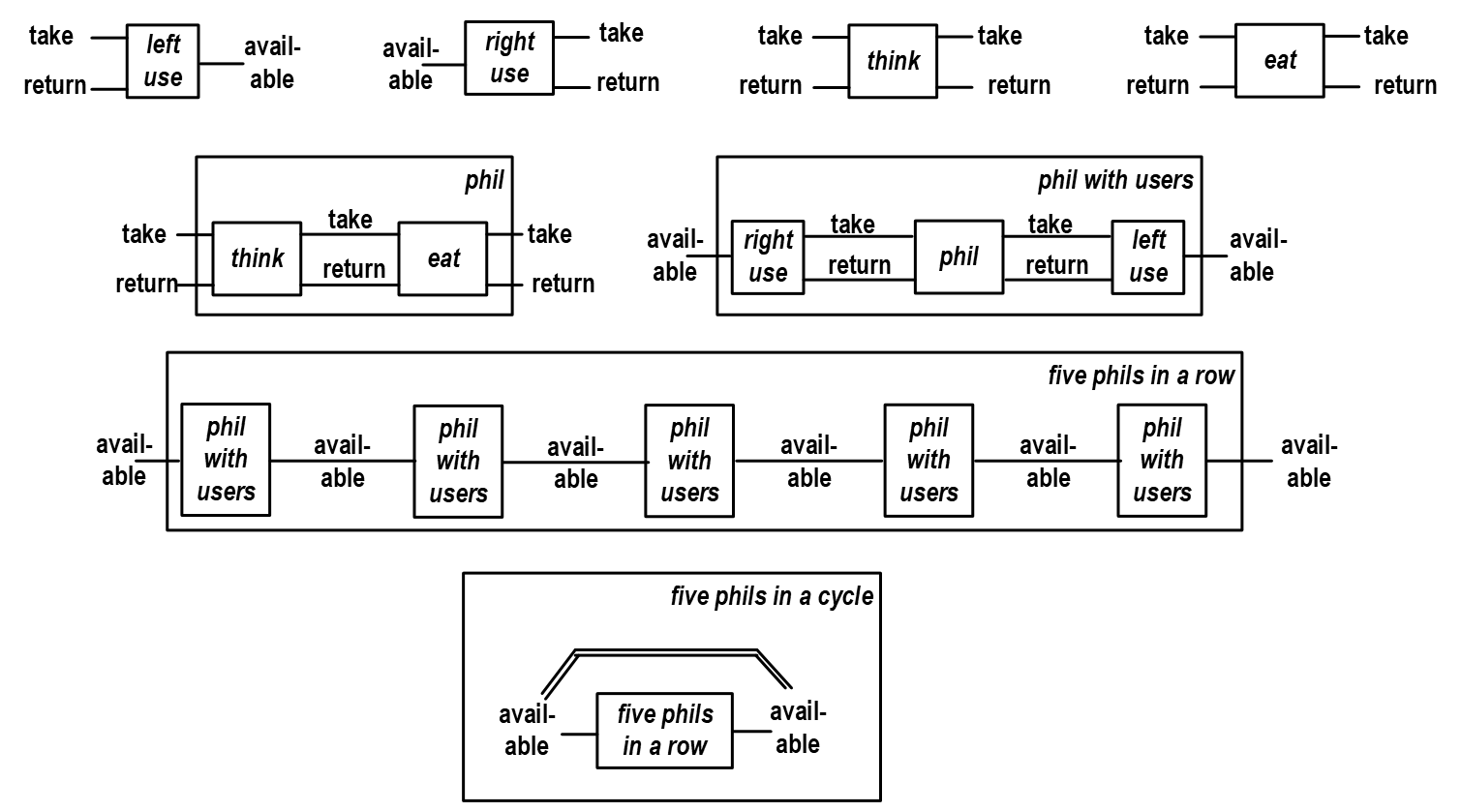}}
\caption{abstract versions of the \textit{phils}-based modules}
\label{fig:abstract_phils}
\end{figure}

\begin{figure}[!tb]
\centering
{\includegraphics[trim={0cm 0cm 0cm 0cm},clip,width=1\textwidth]{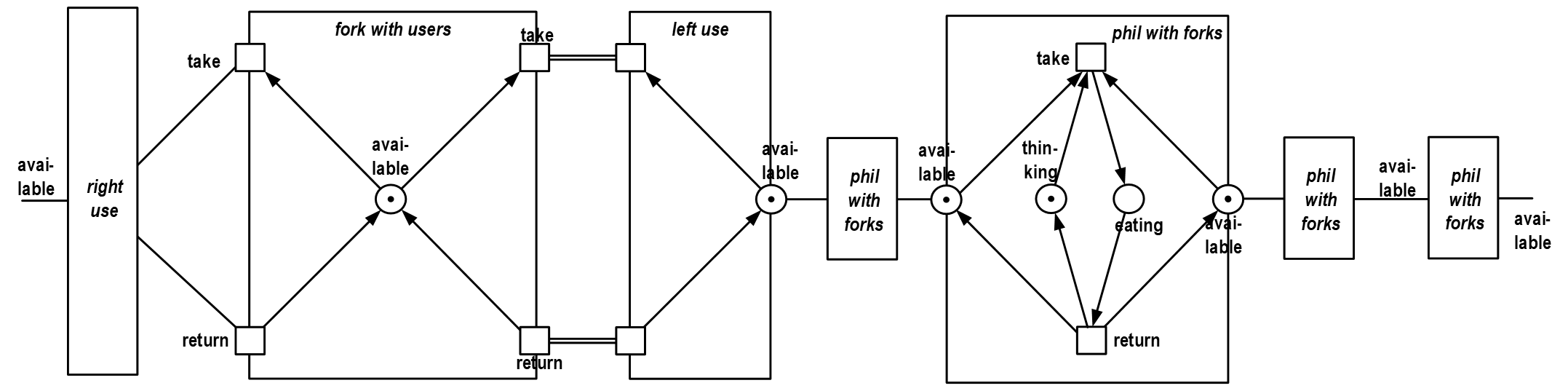}}
\caption{a mixed representation of the philosophers' system}
\label{fig:mixed_version}
\end{figure}

\section{Case Study: A Production Line\label{sec:production_line}}

In the philosophers example, composition of modules is very simple: For all considered compositions $A \bullet B$ of modules $A$ and $B$, the labels of the elements of $A^\ast$ and of $^\ast B$ correspond bijectively, i.e. each element of $A ^\ast$ has an harmonic partner in $^\ast B$, and each element of $^\ast B$ has an harmonic partner in $A ^\ast$. This is not always the case. In fact, modules $A$ and $B$ with \textit{any} labelings of $A^\ast$ and $^\ast B$ can be composed. Hence, for an element of $A ^\ast$ with label $\alpha$ there may be no, or many elements of $^\ast B$ with label $\alpha$. Vice versa, for an element of $^\ast B$ with label $\beta$ there may be no or many element of $A ^\ast$ with label $\beta$. 

To keep matters simple, one may suggest to stick to the case of modules where in each interface, no two elements of an interface are labeled alike.  This, however, is not feasible: composition of two such modules may destroy this property. 

Fig.~\ref{fig:production_and_packing_a} shows a simple production schema: A machine is fed a bunch of material, produces a product, and is ready again for the next bunch of material. Fig.~\ref{fig:production_and_packing_b} extends Fig.~\ref{fig:production_and_packing_a} by a second production step. Of course, this requires a second bunch of material, and produces a second product. 

This example shows that an interface may very well contain equally labeled elements. As we require composition to be a total operation (i.e., any two modules of a given set $M$ of modules can be composed, yielding again a module in $M$), this case cannot be excluded.  

According to the definition of composition, equally labeled elements of an interface are numbered top down. Composition $A \bullet B$ of two modules with equally labeled elements of $A ^\ast$ and $^\ast B$ yields harmonic pairs $\{a, b\}$ of elements $a \in A ^\ast$ and $b \in {^\ast B}$, that turn into inner elements of $A \bullet B$. Remaining elements of $A ^\ast$ and $^\ast B$ go to $(A \bullet B)^\ast$ and $^\ast(A \bullet B)$, respectively. 

The module pack in Fig.~\ref{fig:packing(a)} wraps a product into a parcel.  Fig.~\ref{fig:packing(b)} shows two packings in one module. Fig.~\ref{fig:production_and_packing} shows several combinations of production and packing.

\begin{figure}[!tb]
\centering
\subcaptionbox{module $production$\label{fig:production_line_a}}
{\includegraphics[trim={0cm 0cm 0cm 0cm},clip,scale=0.23]{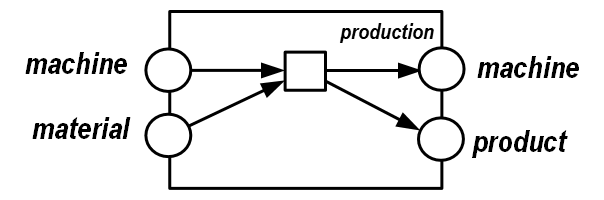}}
\subcaptionbox{module $production \bullet production$\label{fig:production_line_b}}
{\includegraphics[trim={0cm 0cm 0cm 0cm},clip,scale=0.23]{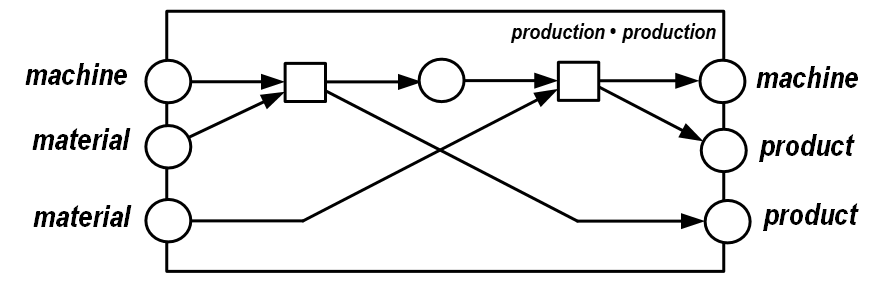}}
\caption{production line}
\label{fig:production_line}
\end{figure}

\begin{figure}[!tb]
\centering
\subcaptionbox{module $pack$\label{fig:packing(a)}}
{\includegraphics[trim={0cm 0cm 0cm 0cm},clip,scale=0.23]{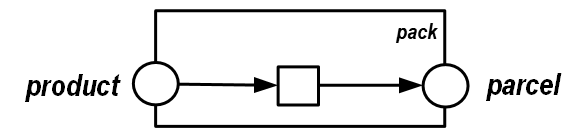}}
\subcaptionbox{module $pack \bullet pack$\label{fig:packing(b)}}
{\includegraphics[trim={0cm 0cm 0cm 0cm},clip,scale=0.23]{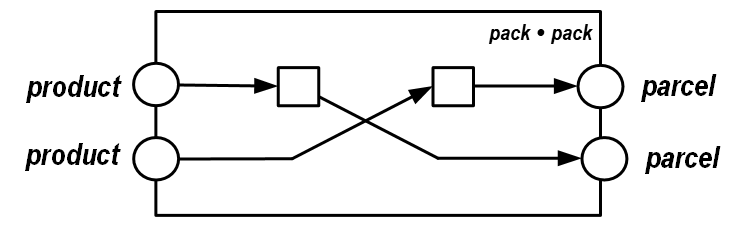}}
\caption{the module $pack$}
\label{fig:packing}
\end{figure}

\begin{figure}[!tb]
\centering
\subcaptionbox{module $production$ and $pack$\label{fig:production_and_packing_a}}
{\includegraphics[trim={0cm 0cm 0cm 0cm},clip,scale=.23]{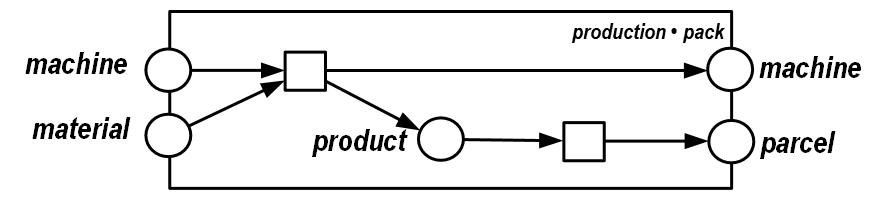}}\vspace{0.5cm}
\subcaptionbox{module $production \bullet pack \bullet production \bullet pack \bullet production$\label{fig:production_and_packing_b}}
{\includegraphics[trim={0cm 0cm 0cm 0cm},clip,scale=.23]{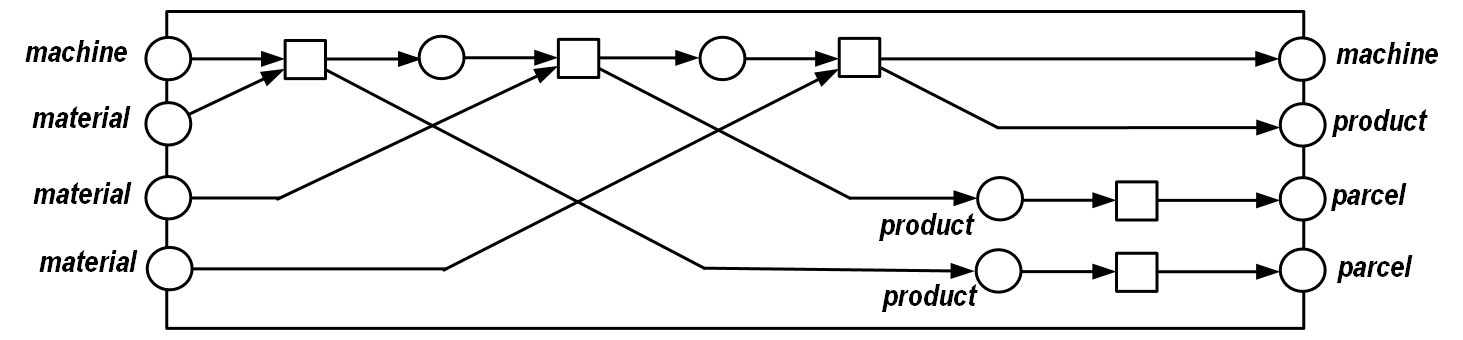}}\vspace{0.5cm}
\subcaptionbox{module $production \bullet production \bullet production \bullet pack \bullet pack$\label{fig:production_and_packing_c}}
{\includegraphics[trim={0cm 0cm 0cm 0cm},clip,scale=.23]{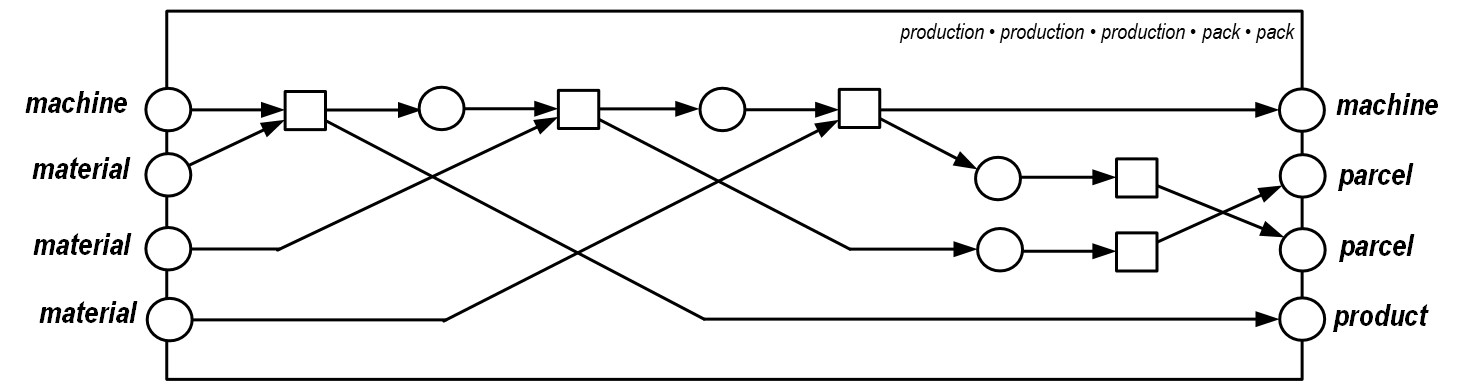}}
\caption{producing three products, and packing two of them}
\label{fig:production_and_packing}
\end{figure}

\section{Related Work\label{sec:related_work}}

The idea of decomposing computer-based systems into submodules is around for a number of decades \cite{parnas1972decomposing}. Numerous proposals and an extensive discussion on the composition of (business) systems in general can be found in literature, e.g. \cite{broy1997refinement,fettke2003specification}. More specific, composition and hierarchies of Petri nets are important and widely considered topics, fundamental for the construction of big net models, and considered since the early 1980ies, e.g. in \cite{suzuki1983method_refinement}. Typical later contributions include the box calculus \cite{best2001pnalgebra} and high-level, tool-supported nets, in particular colored nets \cite{jensen2009colourednets}. Many published concepts for the composition and hierarchical structuring of nets come with very specific properties for specific application areas, e.g. models of service composition.

\section{Conclusions\label{sec:conclusions}}
 
Our proposal just provides a general framework for composition and refinement of nets, applicable to any kind of nets. We suggest an algebraic representation, comparable to the box calculus, complementing the graphical presentation. This is particularly useful for big models, where occasionally a small part may be of interest in detail, but the overall picture should be visible, too. 
It is important that composition forms a monoid on nets, i.e. that any two nets can be composed, yielding again a net, and that composition is associative. This is inevitable in case more than two modules are composed, as the philosophers’ model shows. We are not aware of any composition operator that is total as well as associative. 

The above principles for composition and abstraction are a pillar of the modeling infrastructure \textsc{Heraklit}, where further case studies show usability of these principles \cite{fettke2021modelling,fettke2021handbook}. \textsc{Heraklit} paves the way and provides the means for conducting more field-oriented research on understanding a computer-based system and its decomposition.

%
%
%

\bibliographystyle{splncs04}
\bibliography{main}





\end{document}